\documentclass[12pt]{article}

\usepackage{amssymb}
\usepackage{amsthm}
\newtheorem{theorem}{Theorem}[section]

\newtheorem{definition}[theorem]{Definition}

\newtheorem{remark}[theorem]{Remark}
\usepackage{amssymb,amsmath,bbm}
\usepackage{enumerate}
\usepackage{resizegather}
\usepackage{xcolor}
\usepackage{graphicx}

\def\sqr#1#2{{\vcenter{\vbox{\hrule height .#2pt
     \hbox{\vrule width .#2pt height#1pt \kern#1pt \vrule
     width .#2pt} \hrule height .#2pt}}}}

\newcommand{\norm}[1]{\left\lVert#1\right\rVert}

\newcommand{\R}{{\mathbb{R}}}

\usepackage[]{footmisc}

\DeclareMathOperator{\Den}{Den}

\title{A geometric characterization of VES and Kadiyala-type production functions}
\date{}
\author{Nicol\`o Cangiotti{$^{1}$}, Mattia Sensi{$^{2}$}\\[1em]
\small $^1$University of Pavia, Department of Mathematics, via Ferrata 5, \\ \small 27100 Pavia (PV), Italy. Email: \texttt{nicolo.cangiotti@unipv.it}\\
 \small $^2$University of Trento, Department of Mathematics, via Sommarive 14,\\ \small 38123 Trento (TN), Italy. Email: \texttt{mattia.sensi@unitn.it}}

\begin{document}

\maketitle
\begin{abstract}
The basic concepts of the differential geometry are shortly reviewed and applied to the study of VES production function in the spirit of the works of V\^ilcu and collaborators. A similar characterization is given for a more general production function, namely the Kadiyala production function, in the case of developable surfaces. 
\end{abstract}

\bigskip


\section{Introduction}
\label{Intro}

The study of production functions in the context of neoclassical economy has a long tradition of research from many fields of knowledge.  As pointed out by T. M. Humphrey \cite{Humphrey97}, the first to make a significant contribution to the development of a mathematically consistent approach to the marginal productivity theory was the German mathematical economist, location theorist, and agronomist Johann Heinrich von Th\"unen, in the 19th century (for further details on the history of production functions, see the working paper of S. K. Misha \cite{Mishra2010}). 

The fortune of this mathematical model came in the 1927, when the economist P. Douglas and the mathematics professor C. W. Cobb proposed their famous equation, largely used in the textbooks as well as cited in articles and in  surveys \cite{CobbDouglas28}. Over the following years, the Cobb-Douglas production function became a key concept of neoclassical economics\footnote{It is appropriate to notice that the interpretation of the Cobb-Douglas production function is still a debated topic, as one can read, e.g., in \cite{Labini95}.} (for interesting updates and testing due to Douglas himself, see \cite{Douglas1976}). At the end of the 1950s, R. M. Solow introduced a generalization of the Cobb-Douglas production function: the CES (Constant Elasticity of Substitution) production function \cite{Solow56}; his idea was to aggregate the inputs in a single quantity. The function that realizes this combination of inputs is the so-called \emph{aggregator function}. The aggregator function of CES functions has a constant elasticity of substitution\footnote{We recall that in neoclassical production theory the elasticity of substitution, introduced by J. Hicks in the 1930s \cite{Hiks1932}, provides a measure  degree of substitutability between two factors of productions.}.
However, a different generalization was developed between the 1960s and 1970s by C. A. K. Lovell \cite{Lovell68, Lovell73}, Y. G. Lu and L. Fletcher \cite{Lu68} and N. S. Revanark \cite{Revankar67, Revankar71}: the VES (Variable Elasticity of Substitution) production function. Our analysis stems from the study of this last class of functions (in particular, from the formalization due to Revanark). The approach which we shall use could be called the \emph{differential geometric approach}. This particular technique, in connection with the study of production functions, was introduced and developed much more recently by A. D. V\^ilcu, and G. E. V\^ilcu \cite{Vilcu2011a,Vilcu2011b,Vilcu2013, Vilcu2015, Vilcu2017}. Many contributions are due to B.-Y. Chen \cite{Chen2011,Chen2012a,Chen2012b,Chen2012c,Chen2012d}, and X. Wang \cite{Wang2014, Wang2013}, as well.
The classical theory of  production functions is based on the projections of such functions on a plane, but such an approach does not seem exhaustive, at least from the mathematical point of view. V\^ilcu solved this problem by identifying a production function $Q:\R^n\to \R$ with $\mathcal{Q}$, the graph of $Q$; this turns out to be the nonparametric hypersurface of the $(n+1)$-dimensional Euclidean space $\R^{n+1}$ defined by:
\begin{equation}
G(x_1,\dots,x_n)=\left( x_1,\dots ,x_n, Q(x_1, \dots, x_n) \right ) \qquad (x_1, \dots, x_n) \in \R^n_+.
\end{equation}
Thanks to this reinterpretation, one can study the production functions in terms of the geometry of their graphs $\mathcal{G}\subset \R^{n+1}$. 
\begin{remark}
We are denoting with $G:\R^n\to \R^{n+1}$ a parametrization of the hypersurface $\mathcal{G}$, which is a $n$-dimensional subset of the $(n+1)$-dimensional Euclidean space. This distinction is rather important in the formalism of differential geometry. If $G$ is defined in $A\subset \R^n$, then $\mathcal{G}=G(A)$.
\end{remark}
The aim of the paper is twofold. Starting from basic concepts of differential geometry of surfaces, we shall study the $2$-inputs (i.e. $2$- dimensional) VES production function, obtaining a result, which essentially agree to those achieve by V\^ilcu and collaborators for the generalized Cobb-Douglas production function, and for the generalized CES production function in relation with the Gaussian curvature of the corresponding surface. Consequently, we explore a more general $2$-inputs production function introduced by Kadiyala in the 1970s \cite{Kadiyala72}. 
A renewed interest for the function introduced by Kadiyala seems to have arisen in recent years, particularly due to the works of C. A. Ioan and G. Ioan \cite{Ioan2011} and V\^ilcu \cite{Vilcu2018}. For this particular function, which is a combination of Cobb-Douglas, CES and VES production functions, we prove a result on the corresponding Gaussian curvature in the case of developable surfaces.
\smallskip

The paper is organized as follows. In Section 2 we present an overview of the differential geometry of surfaces, with basic definitions and properties. In Section 3 we study the VES production function as a surface, proving a result which links the returns to scale with the Gaussian curvature (in the same way as done by V\^ilcu, for instance, in \cite{Vilcu2011a}).
In Section 4 we show that the results concerning returns to scale and Gaussian curvature are not valid for a more general 2-inputs production function, namely the Kadiyala production function. 
Finally, in Section 5 we draw the conclusion, giving some suggestions for further developments. 

\section{Basic concepts of Differential Geometry}
\label{Sec.1}
In this section we recall some basic concepts of differential geometry of $2$-dimensional surfaces in $\mathbb{R}^3$ (we refer to \cite{doCarmo} for further readings); these concepts can easily be generalized to $n$-dimensional hypersurfaces in $\mathbb{R}^{n+1}$, for which we refer to \cite{Vilcu2011a}; however, it is unnecessary for the purpose of this article, in which we focus on functions of two variables, namely the VES and the Kadiyala production functions.
\bigskip

Let $U$ be an open set in $\mathbb{R}^2$, and let $f:U\rightarrow \mathbb{R}$ be a (smooth) function. Let $F:\R^2\to \R^3$, defined as
\[
F(x_1,x_2)=(x_1,x_2,f(x_1,x_2)),
\]
be the parametrization of the surface 
\begin{equation}
\mathcal{F}=\{ (x_1,x_2,f(x_1,x_2))\in \mathbb{R}^3 | (x_1,x_2)\in U \}.
\label{eqn:surfa}
\end{equation}
Denote with $\langle\cdot,\cdot\rangle$ the natural inner product on $\mathbb{R}^3$, and with $\norm{\cdot}$ the norm it induces. With this notation, we can give the following:

\begin{definition}
The \emph{first fundamental form} $g$ of the surface $\mathcal{F}$ is given by
\begin{equation}
g:=\sum_{i=1}^{2} g_{ii}\textnormal{d}x_i^2 + 2\sum_{1\leq i<j\leq 2}g_{ij}\textnormal{d}x_i\textnormal{d}x_j=g_{11}\textnormal{d}x_1^2+g_{22}\textnormal{d}x_2^2+2g_{12}\textnormal{d}x_1\textnormal{d}x_2,
\end{equation}
where 
\begin{equation}
g_{ii}=\langle\frac{\partial f}{\partial x_i},\frac{\partial f}{\partial x_i}\rangle, i \in \{ 1,2 \} \quad \textnormal{ and } \quad g_{ij}=\langle\frac{\partial f}{\partial x_i},\frac{\partial f}{\partial x_j}\rangle, 1\leq i<j \leq 2.
\end{equation}
\end{definition}

\begin{definition}
The \emph{second fundamental form} $h$ of the surface $\mathcal{F}$ is given by
\begin{equation}
h:=\sum_{i=1}^{2} h_{ii}\textnormal{d}x_i^2 + 2\sum_{1\leq i<j\leq 2}h_{ij}\textnormal{d}x_i\textnormal{d}x_j=h_{11}\textnormal{d}x_1^2+h_{22}\textnormal{d}x_2^2+2h_{12}\textnormal{d}x_1\textnormal{d}x_2,
\end{equation}
where 
\begin{equation}
h_{ii}=\langle N,\frac{\partial^2 f}{\partial x_i^2}\rangle, i \in \{ 1,2 \} \quad \textnormal{ and } \quad h_{ij}=\langle N,\frac{\partial^2 f}{\partial x_i \partial x_j}\rangle, 1\leq i<j \leq 2,
\end{equation}
and
\begin{equation}
N=\frac{\frac{\partial f}{\partial x_1} \times \frac{\partial f}{\partial x_2}}{\norm{\frac{\partial f}{\partial x_1} \times \frac{\partial f}{\partial x_2}}},
\end{equation}
is the Gauss map of the surface, and $\times$ indicates the vector product in $\mathbb{R}^3$; i.e., $N$ is the unit normal vector of the surface in each point.
\end{definition}
\begin{remark}
The standard notation for the first and second fundamental form is to indicate their (symmetric) matrix representations with the Roman numeral $I=(g_{ij})_{i,j}$ and $II=(h_{ij})_{i,j}$, respectively.
\end{remark}
We can now give the concluding definitions for this section:

\begin{definition}
The \emph{Gaussian curvature} of a point $x$ of the surface is given by
\begin{equation}
K(x)=\frac{\det [II](x)}{\det [I](x)}.
\end{equation}
\end{definition}

\begin{definition}
We call \emph{developable} a surface having zero Gaussian curvature in all its points.
\end{definition}
We are particularly interested in developable surfaces because they can be flattened on a plane by projection, without losing essential information about their geometry, hence easing their study.\\
The main results of the paper are two theorems, which give conditions on the VES and Kadiyala production functions, which ensure the corresponding surfaces are developable.

\section{2-Input VES Production Function}
\label{Sec.2}
This section is devoted to the study of the VES production function, introduced by N. S. Revankar in \cite{Revankar66, Revankar67, Revankar71}:
\begin{equation}
Q(u,v)=k u^{\delta  (1-\beta  \rho )} ((\rho -1) u+v)^{\beta  \delta  \rho }.
\end{equation}
We shall assume the following set of hypotheses:
\begin{equation*}
(\star)
\begin{cases}
k>0, \\ 
0<\beta<1,\\ 
0<\beta \rho <1,\\ 
(\rho -1) u+v>0,\\
\delta>0.
\end{cases}
\end{equation*}
In this settings, $\delta$ is the parameter of return to scale.

\begin{figure}[ht!]
    \centering
\includegraphics[width=0.9\textwidth]{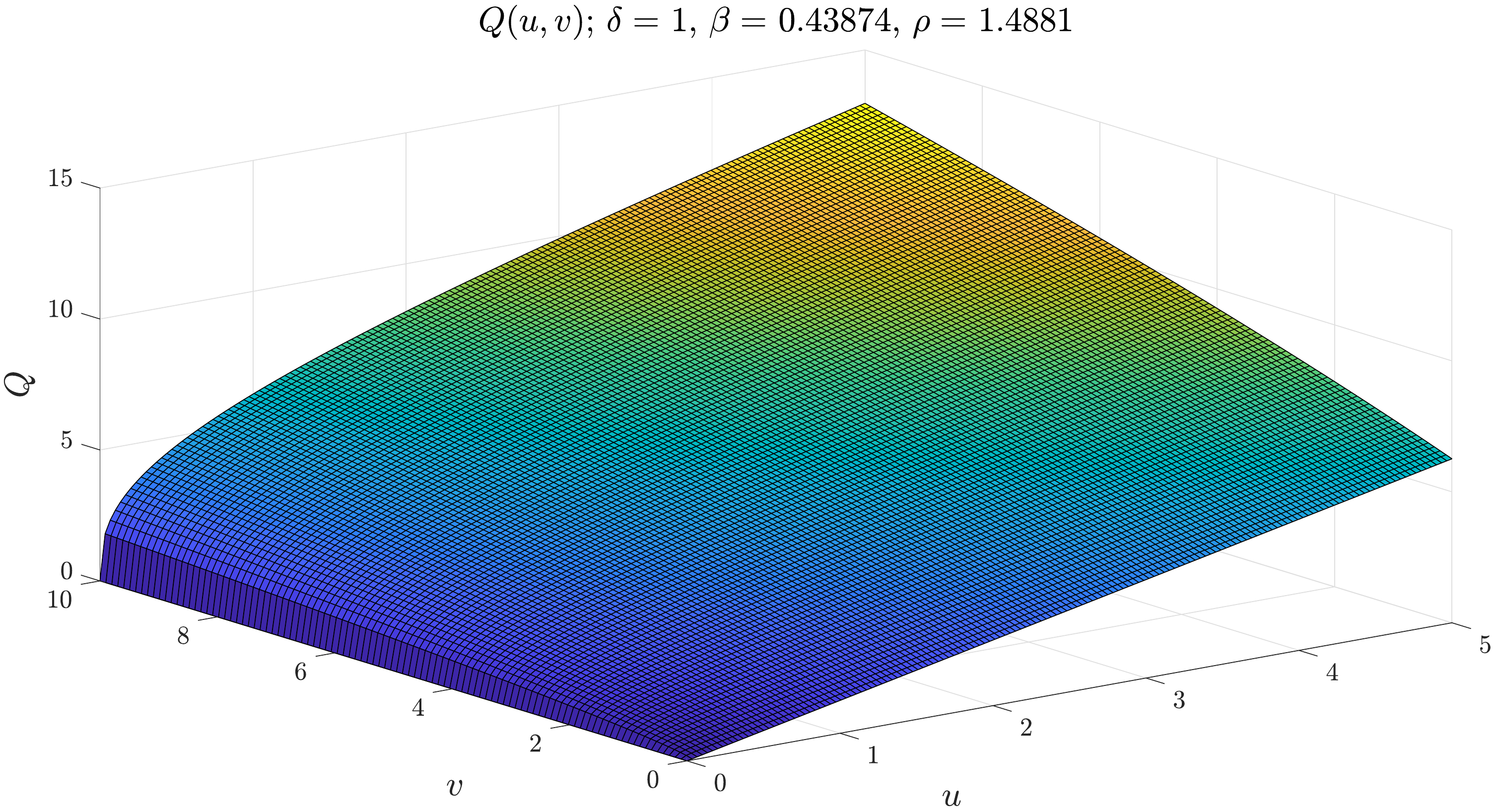}    \caption{Plot of $Q(u,v)$ for a random choice of the parameters $\beta$ and $\rho$ satisfying $(\star)$.}
    \label{fig:Qfig}
\end{figure}

\begin{remark}
We recall that a VES production function has constant, increasing or decreasing returns to scale if $\delta=1$, $\delta>1$, or $\delta<1$, respectively.
\end{remark}
\begin{remark}
The assumptions $(\star)$ allows us to exclude degenerate cases in which, for instance, one of the two inputs is removed. Moreover, the assumption $(\rho-1)u+v>0$ it is necessary when $\rho<1$ to ensure the well-posedness of $Q(u,v)$ (it is clear that if $\rho>1$, since $u,v>0$, this condition is redundant). We notice also that for $\rho <1$ we can rewrite that condition as
\[
\frac{u}{v}<\frac{1}{1-\rho},
\]
or, equivalently,
\begin{equation}
\frac{v}{u}>1-\rho.
\label{eqn:condizz}
\end{equation}
\end{remark}
By using the same notation of Sect. \ref{Intro}, we introduce the following \emph{VES surface} parametrized by
\begin{equation}
\label{Param1}
    G(u,v)=(u,v,Q(u,v)).
\end{equation}
\begin{remark}
In \cite{Revankar66}, Revankar proved that the elasticity of substitution a for the VES production function is
\[
\sigma(u,v)=1+ \frac{\rho-1}{1-\beta\rho}\frac{u}{v}.
\]
Hence, the VES production function varies linearly with the capital-labor ratio $u/v$. In \cite{Revankar71}, Revankar assumes $\sigma>0$ obtaining, as an additional constraint for the economically relevant region of the variables domain,
\[
\frac{v}{u}>\frac{1-\rho}{1-\beta\rho},
\]
which, since $1-\beta\rho<1$, is stricter than (\ref{eqn:condizz}).
\end{remark} 
We can now present and prove the main theorem of this section.
\begin{theorem}
\label{Thm1} 
Let us consider the parametrization of a VES surface defined in Eq. \eqref{Param1}, with $Q(u,v)$ satisfying conditions $(\star)$.

\begin{itemize}
\item The VES production function has constant return to scale if and only if the VES surface is developable.
\item The VES production function has decreasing return to scale if and only if the VES hypersurface has positive Gaussian curvature.
\item The VES production function has increasing return to scale if and only if the VES hypersurface has negative Gaussian curvature.
\end{itemize}
\end{theorem}
\begin{proof}
We can write the Gaussian curvature (defined as in Sec. \ref{Sec.1}) of the VES surface explicitly, using Eq. \eqref{Param1}, obtaining:
\begin{gather}
    K=\frac{\beta  (\delta -1) \delta ^2 k^2 \rho  (\beta  \rho -1) u^{2 (\beta  \delta  \rho +\delta +1)} ((\rho -1) u+v)^{2 \beta  \delta  \rho +2}}{\left(\Den_F(u,v)\right)^2},
\end{gather}
where
\begin{align*}
\Den_F(u,v)=&\delta ^2 k^2 u^{2 \delta } \left(u^2 \left(\rho  \left(\beta ^2 \rho +\rho -2\right)+1\right)-2 (\rho -1) u v (\beta  \rho -1)+\right . \\
& \left . v^2 (\beta  \rho -1)^2\right) ((\rho -1) u+v)^{2 \beta  \delta  \rho }+((\rho -1) u+v)^2 u^{2 \beta  \delta  \rho +2}\\
=&\delta ^2 k^2 u^{2 \delta } \left(\beta ^2 \rho ^2 u^2+((\rho -1) u-v( \rho  \beta -1))^2\right) ((\rho -1) u+v)^{2 \beta  \delta  \rho }+\\
&((\rho -1) u+v)^2 u^{2 \beta  \delta  \rho +2}.
\end{align*}
It is easy to see that $\Den_F(u,v)\neq 0$ for $u,v>0$. The claim follows immediately, keeping in mind assumptions $(\star)$.
\end{proof}

\section{2-Input Kadiyala  Production Function}
In the 1970s, Kadiyala introduced an interesting generalization of production functions \cite{Kadiyala72} (see also the works already mentioned by Ioan \cite{Ioan2011} and V\^ilcu \cite{Vilcu2018}.), which in this section we shall study with a differential geometry approach. The production function is given by:
\begin{equation}
\label{K1}
P(u,v)=\left(k_1 u^{\beta _1+\beta _2}+2 k_2 u^{\beta _1} v^{\beta _2}+k_3 v^{\beta _1+\beta _2}\right){}^{\frac{\delta }{\beta _1+\beta _2}}.
\end{equation}
We shall assume the following set of hypotheses:
\smallskip

\begin{equation*}
(\star\star)
\begin{cases}
k_1 + 2k_2 + k_3 = 1,\\ 
k_i \ge 0, \quad i=1,2,3, \\ 
(k_1,k_2)\neq (0,0),\\
(k_2,k_3) \neq (0,0),\\
\beta_1(\beta_1 + \beta_2)>0,\\
\beta_2(\beta_1 + \beta_2) > 0,\\ 
\delta>0.
\end{cases}
\end{equation*}
As in the previous section, $\delta>0$ is the parameter of returns to scale.
\begin{figure}[ht!]
    \centering
\includegraphics[width=0.9\textwidth]{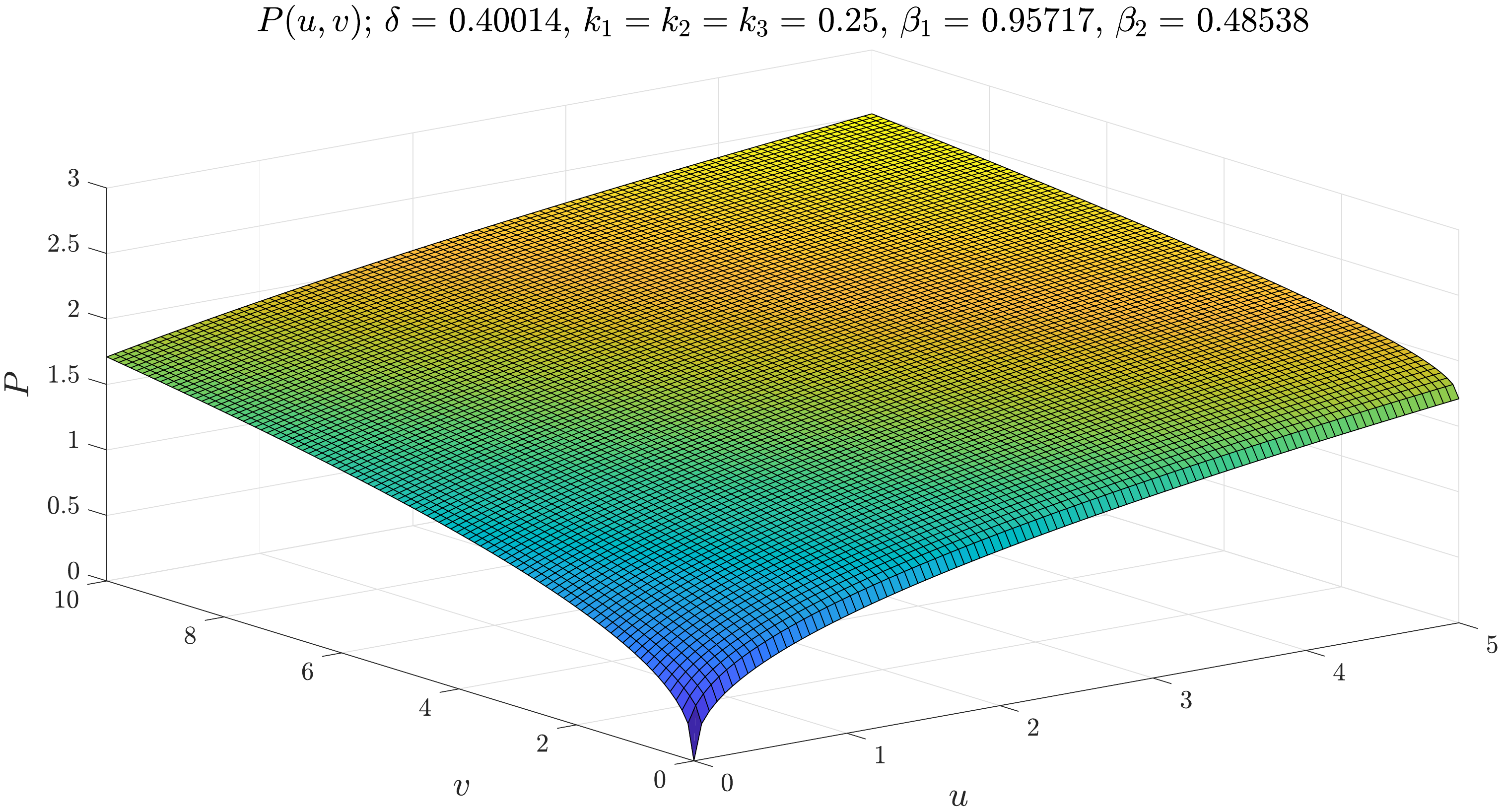}    \caption{Visualization of $P(u,v)$ for a random choice of the parameters $\beta_1$, $\beta_2$ and $\delta<1$.}
    \label{fig:Pmin}
\end{figure}

\begin{figure}[ht!]
    \centering
\includegraphics[width=0.9\textwidth]{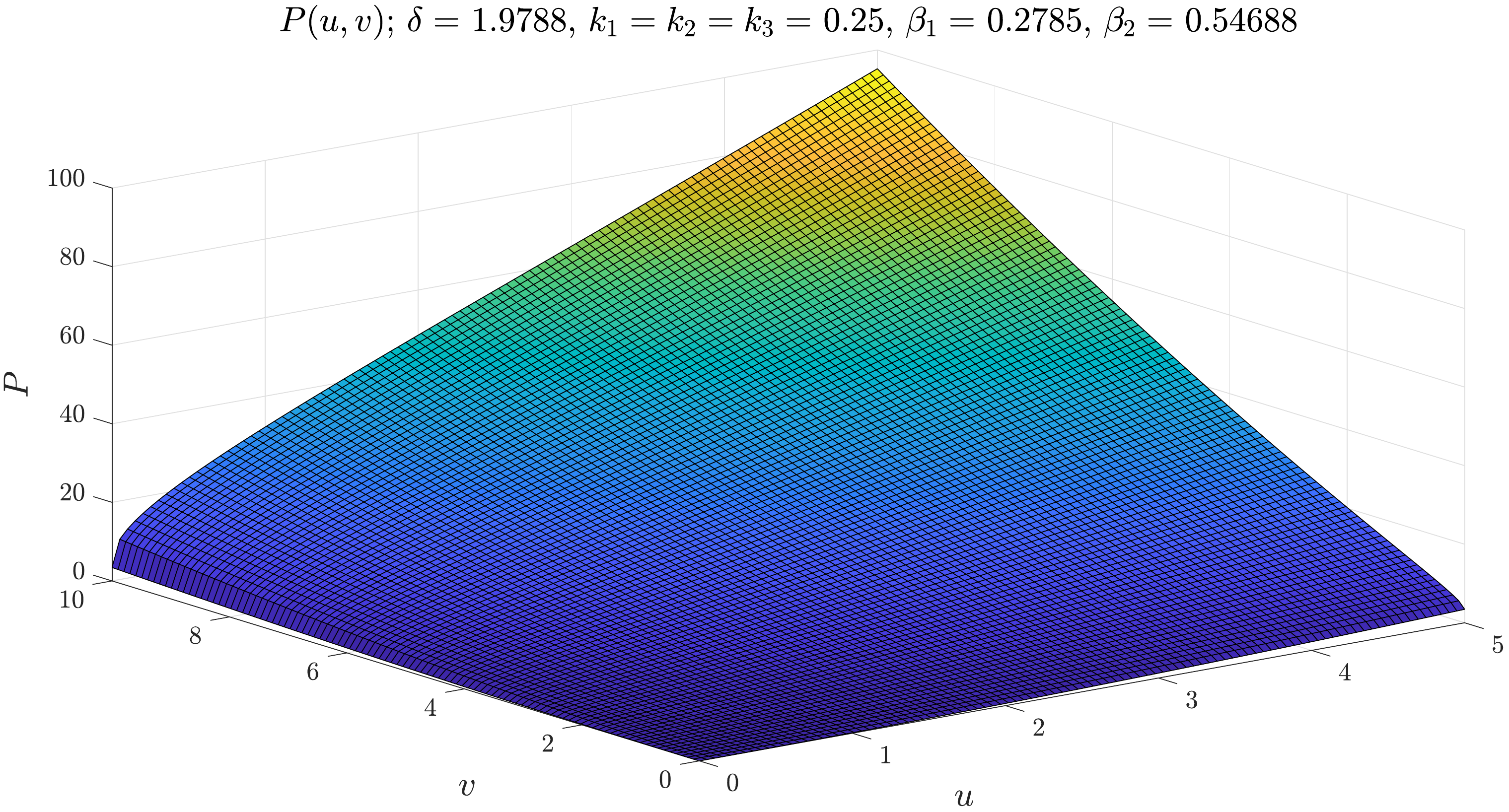}    \caption{Visualization of $P(u,v)$ for a random choice of the parameters $\beta_1$, $\beta_2$ and $\delta>1$.}
    \label{fig:Pmag}
\end{figure}

\begin{remark}
We are assuming $k_1 +2k_2+k_3=1$, without loss of generality. The function $P(u,v)$ is homogeneous of degree one (in the inputs $u$ and $v$) when $\delta=1$ (i.e., for constant returns to scale). We also assume that $\beta_1$ and $\beta_2$ have the same sign as $\beta_1+\beta_2$. In this way, we ensure that the marginal products are non-negative. The two conditions $(k_1,k_2)\neq (0,0)$ and
$(k_2,k_3) \neq (0,0)$ exclude the possibility of the elimination of one input, which would lead to a degenerate production function.
\end{remark}
\begin{remark}
We notice that for $k_2= 0$ we recover a CES-type production function (setting also $\beta_1+\beta_2<1$); for $k_3 = 0$ we obtain the Lu-Fletcher-type production function; for $k_1 = 0$, $k_3 = 0$, and $\delta = 1$ we obtain a Cobb-Douglas-type production function\footnote{We are referring here to the classical Cobb-Douglas function with constant return to scale ($\delta=1$):
\[
Q(u,v)=Au^{1-\alpha}v^{\alpha}, \qquad \alpha=\frac{\beta_1}{\beta_1+\beta_2}.
\]
To obtain increasing/decreasing returns to scale the reader could keep $\delta$ as free parameter.}; finally, for $\beta_1=\frac{1}{\rho \mu}-1, \beta_2 = 1, k_3 = 0$ we get a VES-type production function back.
\end{remark}
By using the same notation as Sect. \ref{Intro}, we introduce the following \linebreak \emph{Kadiyala surface} parametrized by
\begin{equation}
\label{Param2}
    G(u,v)=(u,v,P(u,v)).
\end{equation}

Analogously to the previous section, we can now state the main result. 

\begin{theorem}
Let us consider the Kadiyala surface with the parametrization given by Eq. \eqref{Param2}, with $P(u,v)$ satisfying conditions $(\star\star)$. Then the Kadiyala surface is developable if and only if one of the following conditions holds:
\begin{itemize}
    \item $\delta=1$ (i.e. the Kadiyala production function has constant returns to scale).
    \item $k_2=0$ and $\beta_1+\beta_2=1$
    \item $\beta_1=\beta_2=1$ and $k_2^2-k_1k_3=0$.
\end{itemize}
In particular the last two cases implies that the Kadiyala production function is a \emph{perfect substitutes production function}.
\end{theorem}
\begin{proof}
Firstly, we explicitly calculate the Gaussian curvature of the Kadiyala surface and we obtain:
\begin{equation*}
K=\frac{T_1(u,v)\cdot T_2(u,v)}{(\Den_G(u,v))^2},
\end{equation*}
where
\begin{align*}
    T_1(u,v)=&\left(\beta _1+\beta _2\right){}^2 (\delta -1) \delta ^2 u^{\beta _1+2} v^{\beta _2+2}\cdot\\
    &\cdot\left(k_1 u^{\beta _1+\beta _2}+v^{\beta _2} \left(2 k_2 u^{\beta _1}+k_3 v^{\beta _1}\right)\right){}^{\frac{2 \delta }{\beta _1+\beta _2}+2},\\
    T_2(u,v)=& \left(\beta _1+\beta _2\right) k_1 u^{\beta _2} \Big ( 2 \left(\beta _2-1\right) \beta _2 k_2 u^{\beta _1}+( \beta _1^2+  \\
  &+\left(2 \beta _2-1\right) \beta _1+\left(\beta _2-1\right) \beta _2) k_3 v^{\beta _1}\Big )-2 \beta _1 k_2 v^{\beta _2}\cdot \\
    &\cdot\Big ( 2 \beta _2 k_2 u^{\beta _1}-\left(\beta _1-1\right) \left(\beta _1+\beta _2\right) k_3 v^{\beta _1}\Big ),
\end{align*}
and
\begin{equation*}
\Den_G(u,v)= A_1+A_2+A_3+A_4+A_5,
\end{equation*}
with
\begin{align*}
    A_1=& \left(\beta _1+\beta _2\right){}^2 k_1^2 v^2 u^{2 \left(\beta _1+\beta _2\right)} \cdot \\
    & \cdot \left(\delta ^2 \left(k_1 u^{\beta _1+\beta _2}+v^{\beta _2} \left(2 k_2 u^{\beta _1}+k_3 v^{\beta _1}\right)\right){}^{\frac{2 \delta }{\beta _1+\beta _2}}+u^2\right)\\
    A_2=& \left(\beta _1+\beta _2\right){}^2 k_3^2 u^2 v^{2 (\beta _1+ \beta _2)} \cdot \\
    & \cdot \left(\delta ^2 \left(k_1 u^{\beta _1+\beta _2}+v^{\beta _2} \left(2 k_2 u^{\beta _1}+k_3 v^{\beta _1}\right)\right){}^{\frac{2 \delta }{\beta _1+\beta _2}}+v^2\right)\\
    A_3=& 4 \left(\beta _1+\beta _2\right) k_2 k_3 u^{\beta _1+2} v^{\beta _1+2 \beta _2} \cdot \\
    &\cdot \left(\beta _2 \left(\delta ^2 \left(k_1 u^{\beta _1+\beta _2}+v^{\beta _2} \left(2 k_2 u^{\beta _1}+k_3 v^{\beta _1}\right)\right){}^{\frac{2 \delta }{\beta _1+\beta _2}}+v^2\right)+\beta _1 v^2\right)\\
    A_4=& 4 k_2^2 u^{2 \beta _1} v^{2 \beta _2}\cdot \\ &\cdot \bigg ( \beta _1^2 v^2 \left(\delta ^2 \left(k_1 u^{\beta _1+\beta _2}+v^{\beta _2} \left(2 k_2 u^{\beta _1}+k_3 v^{\beta _1}\right)\right){}^{\frac{2 \delta }{\beta _1+\beta _2}}+u^2\right)+\beta _2^2 u^2 \cdot \\
    &\cdot \left(\delta ^2 \left(k_1 u^{\beta _1+\beta _2}+v^{\beta _2} \left(2 k_2 u^{\beta _1}+k_3 v^{\beta _1}\right)\right){}^{\frac{2 \delta }{\beta _1+\beta _2}}+v^2\right)+2 \beta _1 \beta _2 u^2 v^2 \bigg )\\
    A_5=&2 \left(\beta _1+\beta _2\right) k_1 u^{\beta _1+\beta _2} v^{\beta _2+2} \cdot \\
   &  \cdot  \bigg ( \left(\beta _1+\beta _2\right) k_3 u^2 v^{\beta _1}+2 k_2 u^{\beta _1} \cdot \\
   &\cdot \left(\beta _1 \left(\delta ^2 \left(k_1 u^{\beta _1+\beta _2}+v^{\beta _2} \left(2 k_2 u^{\beta _1}+k_3 v^{\beta _1}\right)\right){}^{\frac{2 \delta }{\beta _1+\beta _2}}+u^2\right)+\beta _2 u^2\right) \bigg )
\end{align*}

$\Den_G(u,v)$ is clearly positive for $(u,v)\neq (0,0)$, since it consists of sums and products of positive terms\footnote{We recall that $(\star\star)$ implies that $\beta_1$, $\beta_1$ and $\beta_1+\beta_2$ have the same sign.}.

To prove the thesis we need to describe the parameter set in which $K\equiv 0$. We remark that
\[ 
K\equiv 0 \iff T_1(u,v)\equiv 0 \vee T_2(u,v)\equiv 0.
\]
If $\delta=1$, we immediately get $T_1(u,v)\equiv 0$, and hence $K\equiv 0$. Let us assume $\delta \neq 1$. In this case $T_1(u,v)\neq 0$, so we can conclude that 
\[
T_1(u,v)\equiv 0\iff \delta=1. 
\]
We shall now study $T_2(u,v)$. Firstly,
we rewrite $T_2(u,v)$, collecting powers of $u$ and $v$ as follows:
\begin{align*}
    T_2(u,v)=&\left(2 \beta _2^3+2 \beta _1 \beta _2^2-2 \beta _2^2-2 \beta _1 \beta _2\right) k_1 k_2 u^{\beta _1+\beta _2}\\
    &-4 \beta _1 \beta _2 k_2^2 u^{\beta _1} v^{\beta _2}\\
    &+\left(\beta _1^3+3 \beta _2 \beta _1^2-\beta _1^2+3 \beta _2^2 \beta _1-2 \beta _2 \beta _1+\beta _2^3-\beta _2^2\right) k_1 k_3 u^{\beta _2} v^{\beta _1}\\
    &+\left(2 \beta _1^3+2 \beta _2 \beta _1^2-2 \beta _1^2-2 \beta _2 \beta _1\right) k_2 k_3 v^{\beta _1+\beta _2}.
\end{align*}
If $k_1= 0$ (or, by symmetry, if $k_3=0$), we obtain $T(u,v)\neq 0$ (in the first quadrant). In the case $k_2=0$, we have
\begin{align*}
    T_2(u,v)=&\left(\beta _1^3+3 \beta _2 \beta _1^2-\beta _1^2+3 \beta _2^2 \beta _1-2 \beta _2 \beta _1+\beta _2^3-\beta _2^2\right) k_1 k_3 u^{\beta _2} v^{\beta _1}\\
    =&\left(\beta _1+\beta _2-1\right) \left(\beta _1+\beta _2\right){}^2k_1 k_3 u^{\beta _2} v^{\beta _1}
\end{align*}
We get $T_2(u,v)\equiv 0$ if and only if\footnote{The solution $\beta_1=-\beta_2$ is forbidden by $(\star\star)$.} $\beta_1+\beta_2=1$. 

Finally, let we assume $k_i\neq 0$ for $i=1,2,3$. If $\beta_1\neq \beta_2$ it is impossible\footnote{Because of, e.g., the term $-4 \beta _1 \beta _2 k_2^2 u^{\beta _1} v^{\beta _2}$.} to obtain $T_2(u,v)=0$. Thus, let us fix $\beta_1=\beta_2=\beta$. So $T_2(u,v)$ becomes:
\begin{align*}
    T_2(u,v) =& 4 \beta^2 \left( \beta- 1 \right) k_1 k_2 u^{2 \beta }+\\
    &4\beta^2 \left(\left(2 \beta - 1 \right) k_1 k_3- k_2^2\right) u^{\beta } v^{\beta }+\\
    &4\beta^2 \left( \beta -1 \right) k_2 k_3 v^{2 \beta },
\end{align*}
which is equal to $0$ if and only if\footnote{The solution $\beta=0$ is forbidden by $(\star\star)$.} $\beta=1$ and $k_1k_3=k_2^2$. The proof is completed by noting that
\[
T_2(u,v)\equiv 0 \iff \left ( k_2=0 \wedge \beta_1+\beta_2=1 \right ) \vee \left ( \beta_1=\beta_2=1 \wedge k_1k_3=k_2^2 \right ).
\]
In conclusion, we notice that for $k_2=0$ and $\beta_1+\beta_2=1$ we have the following function:
\[
P_1(u,v)=(k_1u+k_3v)^{\delta}.
\]
Moreover, for $\beta_1=\beta_2=1$ and $k_2^2-k_1k_3=0$ we obtain
\[
P_2(u,v)=\left( \sqrt{k_1}u+\sqrt{k_3}v\right)^{\delta}.
\]
Thus, both cases lead to a perfect substitutes production function.
\end{proof}
\section{Summary and conclusions}
In this paper we analyze two production functions from the point of view of differential geometry. 

In particular, in accordance with the approach of V\^ilcu, we give a characterization of the (2-input) VES function in terms of curvature of the related surface. This result is analogous (as we expected) to the results obtained by V\^ilcu for the Cobb-Douglas and the CES production function. 

The second part of the paper is devoted to another kind of production function, which could be seen as a combination of the most famous (2-inputs) production functions. We call it Kadyiala production function. For the latter, computations become more cumbersome, but it is still possible to give a characterization connected with the Gaussian curvature of the corresponding surface, at least in the case of developable surfaces. The constant returns to scale is a necessary condition if we suppose that the Kadyiala production function is not a perfect substitutes production function; this result is consistent with the previous works of V\^ilcu, as well.

We conclude with a short outlook on possible research perspectives: a natural successive step in our analysis would be to study in detail the sign of the curvature of the Kadyiala production function, its dependence on specific choices of the parameters and the interpretation of such picks. Another logical path to follow would be to generalize the results presented in this paper for functions of a generic number of inputs $n$; however, one would need to propose a clever way of analyzing such a function, since computations proved to be cumbersome even for the 2-dimensional case; in this regard, it would be particularly interesting to study the connections between our work and recent papers by Ioan \cite{Ioan2011} and V\^ilcu \cite{Vilcu2018}. 
\par\bigskip\noindent
{\bf Declarations of interest:} none.
\par\bigskip\noindent
{\bf Acknowledgments:} NC and MS would like to thank the University of Pavia and the University of Trento, respectively, for supporting their research.

\bibliographystyle{acm} \small
\bibliography{biblio}

\begin{thebibliography}{10}

\bibitem{Chen2011}
{\sc Chen, B.-Y.}
\newblock On some geometric properties of $ h $-homogeneous production
  functions in microeconomics.
\newblock {\em Kragujev. J. Math. 35}, 37 (2011), 343--357.

\bibitem{Chen2012a}
{\sc Chen, B.-Y.}
\newblock Classification of $ h $-homogeneous production functions with
  constant elasticity of substitution.
\newblock {\em Tamkang J. Math. 43}, 2 (2012), 321--328.

\bibitem{Chen2012d}
{\sc Chen, B.-Y.}
\newblock Geometry of quasi-sum production functions with constant elasticity
  of substitution property.
\newblock {\em Journal of Advanced Mathematical Studies 5}, 2 (2012), 90--98.

\bibitem{Chen2012b}
{\sc Chen, B.-Y.}
\newblock A note on homogeneous production models.
\newblock {\em Kragujev. J. Math. 36}, 38 (2012), 41--43.

\bibitem{Chen2012c}
{\sc Chen, B.-Y.}
\newblock On some geometric properties of quasi-sum production models.
\newblock {\em J. Math. Anal. Appl. 392}, 2 (2012), 192--199.

\bibitem{CobbDouglas28}
{\sc Cobb, C.~W., and Douglas, P.~H.}
\newblock A theory of production.
\newblock {\em Am. Econ. Rev. 18}, 1 (1928), 139--165.

\bibitem{Douglas1976}
{\sc Douglas, P.~H.}
\newblock The {C}obb-{D}ouglas production function once again: its history, its
  testing, and some new empirical values.
\newblock {\em J. Polit. Econ. 84}, 5 (1976), 903--915.

\bibitem{Hiks1932}
{\sc Hiks, J.~R.}
\newblock Theory of wages, 1932.

\bibitem{Humphrey97}
{\sc Humphrey, T.~M.}
\newblock Algebraic production functions and their uses before
  {C}obb-{D}ouglas.
\newblock {\em FRB Richmond Economic Quarterly 83}, 1 (1997), 51--83.

\bibitem{Ioan2011}
{\sc Ioan, C.~A., and Ioan, G.}
\newblock A generalization of a class of production functions.
\newblock {\em Appl. Econ. Lett. 18}, 18 (2011), 1777--1784.

\bibitem{Kadiyala72}
{\sc Kadiyala, K.~R.}
\newblock Production functions and elasticity of substitution.
\newblock {\em South. Econ. J.\/} (1972), 281--284.

\bibitem{Labini95}
{\sc Labini, P.~S.}
\newblock Why the interpretation of the {C}obb-{D}ouglas production function
  must be radically changed.
\newblock {\em Struct. Chang. Econ. Dyn. 6}, 4 (1995), 485--504.

\bibitem{Lovell68}
{\sc Lovell, C. A.~K.}
\newblock Capacity utilization and production function estimation in postwar
  {A}merican manufacturing.
\newblock {\em Q. J. Econ.\/} (1968), 219--239.

\bibitem{Lovell73}
{\sc Lovell, C. A.~K.}
\newblock Estimation and prediction with {CES} and {VES} production functions.
\newblock {\em Int. Econ. Rev. (Philadelphia)\/} (1973), 676--692.

\bibitem{Lu68}
{\sc Lu, Y.-C., and Fletcher, L.~B.}
\newblock A generalization of the {CES} production function.
\newblock {\em Rev. Econ. Stat.\/} (1968), 449--452.

\bibitem{Mishra2010}
{\sc Mishra, S., et~al.}
\newblock A brief history of production functions.
\newblock {\em The IUP Journal of Managerial Economics 8}, 4 (2010), 6--34.

\bibitem{doCarmo}
{\sc Perdig{\~a}o~do Carmo, M.}
\newblock Differential geometry of curves and surfaces, 1976.

\bibitem{Revankar66}
{\sc Revankar, N.~S.}
\newblock The constant and variable elasticity of substitution production
  functions: a comparative study in {U.S.} manufacturing, mimeographed, systems
  formulation.
\newblock {\em Methodology and Policy Workshop Paper 6603, Social Systems
  Research Institute, University of Wisconsin\/} (1966).

\bibitem{Revankar67}
{\sc Revankar, N.~S.}
\newblock Production functions with variable elasticity of substitution and
  variable returns to scale.
\newblock {\em Doctoral Dissertation (unpublished), University of Wisconsin\/}
  (1968).

\bibitem{Revankar71}
{\sc Revankar, N.~S.}
\newblock A class of variable elasticity of substitution production functions.
\newblock {\em Econometrica\/} (1971), 61--71.

\bibitem{Solow56}
{\sc Solow, R.~M.}
\newblock A contribution to the theory of economic growth.
\newblock {\em Q. J. Econ. 70}, 1 (1956), 65--94.

\bibitem{Vilcu2011b}
{\sc V{\^\i}lcu, A.~D., and V{\^\i}lcu, G.~E.}
\newblock On some geometric properties of the generalized {CES} production
  functions.
\newblock {\em Appl. Math. Comput. 218}, 1 (2011), 124--129.

\bibitem{Vilcu2013}
{\sc V{\^\i}lcu, A.~D., and V{\^\i}lcu, G.~E.}
\newblock On homogeneous production functions with proportional marginal rate
  of substitution.
\newblock {\em Math. Probl. Eng. 2013\/} (2013).

\bibitem{Vilcu2015}
{\sc V{\^\i}lcu, A.~D., and V{\^\i}lcu, G.~E.}
\newblock Some characterizations of the quasi-sum production models with
  proportional marginal rate of substitution.
\newblock {\em Comptes Rendus Math. 353}, 12 (2015), 1129--1133.

\bibitem{Vilcu2017}
{\sc V{\^\i}lcu, A.~D., and V{\^\i}lcu, G.~E.}
\newblock A survey on the geometry of production models in economics.
\newblock {\em Arab. J. Math. Sci. 23\/} (2017), 18--31.

\bibitem{Vilcu2011a}
{\sc V{\^\i}lcu, G.~E.}
\newblock A geometric perspective on the generalized {C}obb--{D}ouglas
  production functions.
\newblock {\em Appl. Math. Lett. 24}, 5 (2011), 777--783.

\bibitem{Vilcu2018}
{\sc V{\^\i}lcu, G.-E.}
\newblock On a generalization of a class of production functions.
\newblock {\em Appl. Econ. Lett. 25}, 2 (2018), 106--110.

\bibitem{Wang2014}
{\sc Wang, X.}
\newblock A characterization of {CES} production functions having constant
  return to scale in microeconomics.
\newblock {\em Int. J. Appl. Math. Stat. 52}, 2 (2014), 152--158.

\bibitem{Wang2013}
{\sc Wang, X., and Fu, Y.}
\newblock Some characterizations of the {C}obb-{D}ouglas and {CES} production
  functions in microeconomics.
\newblock In {\em Abstr. Appl. Anal.\/} (2013), vol.~2013, Hindawi.

\end{thebibliography}

\end{document}